%% file: main.tex
\documentclass[runningheads,envcountsame,envcountsect]{llncs}

\usepackage{tikz}

\usepackage{amsmath,amssymb}
\usepackage{graphicx}
\usepackage{hyperref}

\let\doendproof\endproof
\renewcommand\endproof{~\hfill$\qed$\doendproof}

\newcommand{\mso}[1]{$\mathsf{MSO}_{#1}$}
\newcommand{\fo}{$\mathsf{FO}$}

\usepackage{color}
\definecolor{darkred}{rgb}{0.7,0,0}

\newcommand{\mw}{\mathtt{mw}}  

\newcommand{\OPT}{A}

\newenvironment{listing}[1]{%
  \begin{list}{*}{%
      \settowidth{\labelwidth}{#1}%
      \setlength{\leftmargin}{\labelwidth}%
      \advance \leftmargin by 5pt
      \setlength{\itemsep}{0pt}%
      \setlength{\parsep}{0pt}%
      \setlength{\topsep}{0pt}%
      \setlength{\parskip}{0pt}%
    }%
  }{%
  \end{list}
}

\makeatletter
\usepackage{tabularx,environ}
\makeatletter
\newcommand{\problemtitle}[1]{\gdef\@problemtitle{#1}}
\newcommand{\probleminput}[1]{\gdef\@probleminput{#1}}
\newcommand{\problemquestion}[1]{\gdef\@problemquestion{#1}}
\NewEnviron{myproblem}{
  \problemtitle{}\probleminput{}\problemquestion{}
  \BODY
  \par\addvspace{.5\baselineskip}
  \centering
  \fbox{
    \parbox{0.96\hsize}{
      \@problemtitle
      \begin{listing}{\textbf{Question:}}
	\item[\textbf{Input:}] \@probleminput
	\item[\textbf{Question:}] \@problemquestion
      \end{listing}
    }
  }
  \par\addvspace{.5\baselineskip}
}
\makeatother

\begin{document}
\title{Grouped Domination Parameterized by Vertex Cover, Twin Cover, and Beyond\texorpdfstring{\thanks{%
Partially supported
by JSPS KAKENHI Grant Numbers 
JP17H01698, 
JP17K19960, 
JP18H04091, 
JP20H05793, 
JP20H05967, 
JP21K11752, 
JP21H05852, 
JP21K17707, 
JP21K19765, 
and JP22H00513. 
}}{}}
\titlerunning{Grouped Domination Parameterized by Structural Parameters}
%
\author{Tesshu Hanaka\inst{1}\orcidID{0000-0001-6943-856X} \and
Hirotaka Ono\inst{2}\orcidID{0000-0003-0845-3947} \and
Yota Otachi\inst{2}\orcidID{0000-0002-0087-853X} \and 
Saeki Uda\inst{2}}
\authorrunning{T. Hanaka et al.}
%
\institute{Kyushu University, Fukuoka, Japan \email{hanaka@inf.kyushu-u.ac.jp} \and
Nagoya University, Nagoya, Japan \email{ono@nagoya-u.jp, otachi@nagoya-u.jp, uda.saeki.z4@s.mail.nagoya-u.ac.jp}}
\maketitle              
\begin{abstract}
A dominating set $S$ of graph $G$ is called an \emph{$r$-grouped dominating set} if $S$ can be partitioned into $S_1,S_2,\ldots,S_k$ such that the size of each unit $S_i$ is $r$ and the subgraph of $G$ induced by $S_i$ is connected. The concept of $r$-grouped dominating sets generalizes several well-studied variants of dominating sets with requirements for connected component sizes, such as the ordinary dominating sets ($r=1$), paired dominating sets ($r=2$), and connected dominating sets ($r$ is arbitrary and $k=1$). 
In this paper, we investigate the computational complexity of \textsc{$r$-Grouped Dominating Set}, which is the problem of deciding whether a given graph has an $r$-grouped dominating set with at most $k$ units. For general $r$, \textsc{$r$-Grouped Dominating Set} is hard to solve in various senses because the hardness of the connected dominating set is inherited. We thus focus on the case in which  $r$ is a constant or a parameter, but we see that \textsc{$r$-Grouped Dominating Set} for every fixed $r>0$ is still hard to solve. From the observations about the hardness, we consider the parameterized complexity concerning well-studied graph structural parameters. 
We first see that \textsc{$r$-Grouped Dominating Set} is fixed-parameter tractable for $r$ and treewidth, 
which is derived from the fact that the condition of $r$-grouped domination for a constant $r$ can be represented as monadic second-order logic (\mso{2}).
This fixed-parameter tractability is good news, but the running time is not practical.  
We then design an $O^*(\min\{(2\tau(r+1))^{\tau},(2\tau)^{2\tau}\})$-time algorithm for general $r\ge 2$, where $\tau$ is the twin cover number, which is a parameter between vertex cover number and clique-width. 
For paired dominating set and trio dominating set, i.e., $r \in \{2,3\}$, we can speed up the algorithm, whose running time becomes $O^*((r+1)^\tau)$. We further argue the relationship between FPT results and graph parameters, which draws the parameterized complexity landscape of \textsc{$r$-Grouped Dominating Set}. 

\keywords{Dominating Set \and Paired Dominating Set \and Parameterized Complexity \and Graph Structural Parameters.}
\end{abstract}
\section{Introduction}
\subsection{Definition and motivation}\label{sec:intro1}
Given an undirected graph $G=(V,E)$, a vertex set $S\subseteq V$ is called a \emph{dominating set} if every vertex in $V$ is either in $S$ or adjacent to a vertex in $S$. The dominating set problem is the problem of finding a dominating set with the minimum cardinality. Since the definition of dominating set, i.e., covering all the vertices via edges, is natural, many practical and theoretical problems are modeled as dominating set problems with additional requirements;  many variants of dominating set are considered and investigated. Such variants somewhat generalize or extend the ordinary dominating set based on theoretical or applicational motivations. 
In this paper, we focus on variants that require the dominating set to satisfy specific connectivity and size constraints. 
One example considering connectivity is the connected dominating set. A dominating set is called a \emph{connected dominating set} if the subgraph induced by a dominating set is connected. Another example is the paired dominating set. 
A paired dominating set is a dominating set of a graph such that the subgraph induced by it admits a perfect matching.

This paper introduces the $r$-grouped dominating set, which generalizes the connected dominating set,  the paired dominating set, and some other variants. 
A dominating set $S$ is called an \emph{$r$-grouped dominating set} if $S$ can be partitioned into $\{S_1,S_2,\ldots,S_k\}$ such that each $S_i$ is a set of $r$ vertices and $G[S_i]$ is connected. We call each $S_i$ a \emph{unit}. The $r$-grouped dominating set generalizes both the connecting dominating set and the paired dominating set in the following sense: a connecting dominating set with $r$ vertices is equivalent to an $r$-grouped dominating set of one unit, and a paired dominating set with $k$ pairs is equivalent to a $2$-grouped dominating set with $k$ units. 

This paper investigates the parameterized complexity of deciding whether a given graph has an $r$-grouped dominating set with $k$ units. The parameters that we focus on are so-called graph structural parameters, such as vertex cover number and twin-cover number.
The results obtained in this paper are summarized in Our Contribution (Section \ref{sec:contribution}). 

\subsection{Related work}
An enormous number of papers study the dominating set problem, including the ones strongly related to the $r$-grouped dominating set. 

The dominating set problem is one of the most important graph optimization problems.
Due to its NP-hardness, its tractability is finely studied from several aspects, such as approximation, solvable graph classes, fast exact exponential-time solvability, and parameterized complexity. Concerning the parameterized complexity, the dominating set problem is W[2]-complete for solution size $k$; it is unlikely to be fixed-parameter tractable~\cite{CyganFKLMPPS15}. 
On the other hand, since the dominating set can be expressed in \mso{1}, it is FPT when parametrized by clique-width or treewidth (see, e.g., \cite{Kreutzer11}).

The connected dominating set is a well-studied variant of dominating set. This problem arises in communication and computer networks such as mobile ad hoc networks. It is also W[2]-hard when parameterized by the solution size \cite{CyganFKLMPPS15}. Furthermore, the connected dominating set also can be expressed in \mso{1};  it is FPT when parametrized by clique-width and treewidth as in the ordinary dominating set problem.
Furthermore, single exponential-time algorithms for connected dominating set parameterized by treewidth can be obtained by the Cut \& Count technique \cite{CNPMMVW2022} or the rank-based approach \cite{BCKN2015}.

The notion of the paired dominating set is introduced in \cite{HS1995:paired,HS1998:paired} by Haynes and Slater  as a model of dominating sets with pairwise backup. 
It is NP-hard on split graphs, bipartite graphs~\cite{CLZ2010:paierd}, graphs of maximum degree 3~\cite{CLZ2009:paired:approx}, and planar graphs of maximum degree 5 \cite{TripathiKPPW22},  whereas it can be solved in polynomial time on strongly-chordal graphs~\cite{CLZ2009:paierd}, distance-hereditary graphs~\cite{LKH2020:paired}, and AT-free graphs~\cite{TripathiKPPW22}.  There are several graph classes (e.g., strongly orderable graphs \cite{PP2019:paired}) where the paired dominating set problem is tractable, whereas the ordinary dominating set problem remains NP-hard. For other results about the paired dominating set, see a survey~\cite{Desormeaux2020}.

\subsection{Our contributions}\label{sec:contribution}
This paper provides a unified view of the parameterized complexity of dominating set problem variants with connectivity and size constraints.   

As mentioned above, an $r$-grouped dominating set of $G$ with $1$ unit is equivalent to a connected dominating set with size $r$, which implies that some hardness results of \textsc{$r$-Grouped Dominating Set} for general $r$ are inherited directly from \textsc{Connected Dominating Set}. From these, we mainly consider the case where $r$ is a constant or a parameter. 

Unfortunately, \textsc{$r$-Grouped Dominating Set} for $r=1,2$ is also hard to solve again because \textsc{$1$-Grouped Dominating Set} and 
\textsc{$2$-Grouped Dominating Set} are respectively the ordinary dominating set problem and the paired dominating set problem. 
Thus, it is worth considering whether a larger but constant $r$ enlarges, restricts, or leaves unchanged the graph classes for which similar hardness results hold. A way to classify or characterize graphs of certain classes is to focus on graph-structural parameters. 
By observing that the condition of $r$-grouped dominating set can be represented as monadic second-order logic (\mso{2}), we can see that \textsc{$r$-Grouped Dominating Set} is fixed-parameter tractable for $r$ and treewidth. Recall that the condition of the connected dominating set can be represented as monadic second-order logic (\mso{1}), which implies that there might exist a gap between $r=1$ and $2$, or between $k=1$ and $k>1$. 
Although this FPT result is good news, its time complexity is not practical. From these observations, we focus on less generalized graph structural parameters, vertex cover number $\nu$ or twin cover number $\tau$ as a parameter, and design single exponential fixed-parameter algorithms for \textsc{$r$-Grouped Dominating Set}. 

Our algorithm is based on dynamic programming on nested partitions of a vertex cover, and its running time is $O^*(\min\{(2\nu(r+1))^{\nu},(2\nu)^{2\nu}\})$ for general $r\ge 2$.
For paired dominating set and trio dominating set, i.e., $r \in \{2,3\}$, we can tailor the algorithm to run in $O^*((r+1)^\nu)$ time by observing that the nested partitions of a vertex cover degenerate in some sense. 

We then turn our attention to a more general parameter, the twin cover number.
We show that, given a twin cover, \textsc{$r$-Grouped Dominating Set} admits an optimal solution 
in which twin-edges do not contribute to the connectivity of $r$-units.
This observation implies that these edges can be removed from the graph, and thus we can focus on the resultant graph of bounded vertex cover number.
Hence, we can conclude that our algorithms still work when the parameter $\nu$ in the running time is replaced with twin cover number $\tau$.

We further argue the relationship between FPT results and graph parameters. The perspective is summarized in Figure \ref{fig:parameters0}, which draws the parameterized complexity landscape of \textsc{$r$-Grouped Dominating Set}.

\begin{figure}[tbh]
  \centering
  \input{parameters}
  \caption{The complexity of $r$-\textsc{Grouped Dominating Set} with respect to structural graph parameters.
    An edge between two parameters indicates that there is a function in the one above 
    that lower-bounds the one below (e.g., $\text{treewidth} \le \text{pathwidth}$).}
  \label{fig:parameters0}
\end{figure}
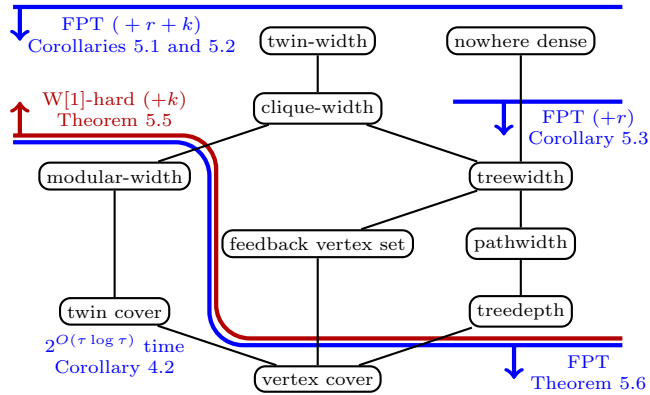



\section{Preliminaries}
Let $G=(V,E)$ be an undirected graph. For a vertex subset $V'\subseteq V$, the subgraph induced by $V'$ is denoted by $G[V']$. Also, let us denote by $N(v)$ and $N[v]$ the open neighborhood and the closed neighborhood of $v$, respectively. The degree  of a vertex $v$ is defined by $d(v) = |N(v)|$. The maximum degree of $G$ is denoted by $\Delta$.

A vertex set $S$ is a \emph{vertex cover} of $G$ if for every edge $\{u,v\}\in E$, at least one of $u,v$ is in $S$.
The \emph{vertex cover number} $\nu$ of $G$ is defined by the size of a minimum vertex cover of $G$. A minimum vertex cover of $G$ can be found in $O^*(1.2738^{\nu})$ time~\cite{CKX2010}.\footnote{The $O^*$ notation suppresses the polynomial factors of the input size.}

Two vertices $u$ and $v$ are (\emph{true}) \emph{twins} if $N[u]=N[v]$.
An edge  $\{u,v\}\in E$ is a \emph{twin edge} if $u$ and $v$ are true twins.
A vertex set $S$ is a \emph{twin cover} if for every edge $\{u,v\}\in E$, either $\{u,v\}$ is a twin edge, or at least one of $u,v$ is in $S$.
The size $\tau$ of a minimum twin cover of $G$ is called the \emph{twin cover number} of $G$. A minimum twin cover of $G$ can be found in $O^*(1.2738^{\tau})$ time~\cite{Ganian2015}.


We briefly introduce basic terminology of parameterized complexity. Given an input size $n$ and a parameter $k$,  a problem is \emph{fixed-parameter tractable (FPT)} if it can be solved in $f(k)n^{O(1)}$ time where $f$ is some computable function. Also,  a problem is \emph{slice-wise polynomial (XP)} if it can be solved in $n^{f(k)}$ time. 
See standard textbooks (e.g., \cite{CyganFKLMPPS15}) for more details.


\subsection{$r$-Grouped Dominating Set}
An \emph{$r$-grouped dominating set with $k$ units} in $G$ is a family $\mathcal{D} = \{D_1, \ldots, D_k\}$ of subsets of $V$  such that $D_i$'s are mutually disjoint, $|D_i|=r$, $G[D_i]$ is connected for $1\le i\le k$, and $\bigcup_{D\in \mathcal{D}} D$ is a dominating set of $G$. For simplicity, let  $\bigcup \mathcal{D}$ denote $\bigcup_{D\in \mathcal{D}} D$. 
We say that $\mathcal{D}$ is a minimum $r$-grouped dominating set if it is an $r$-grouped dominating set with the minimum number of units.

\begin{myproblem}
  \problemtitle{$r$-\textsc{Grouped Dominating Set}}
  \probleminput{A graph $G$ and positive integers $r$ and $k$.}
  \problemquestion{Is there an $r$-grouped dominating set with at most $k$ units in $G$?}
\end{myproblem}

\section{Basic Results}
In this section, we prove $r$-\textsc{Grouped Dominating Set} is W[2]-hard but XP when parameterized by $k+r$ and it is NP-hard even on planar bipartite graphs of maximum degree 3.

We first observe that finding an $r$-grouped dominating set with at most $1$ unit is equivalent to finding a connected dominating set of size $r$.
Thus, the W[2]-hardness of \textsc{$r$-Grouped Dominating Set} parameterized by $r$ follows the one of \textsc{Connected Dominating Set} parameterized by the solution size.
Also,  the case $r = 1$ follows immediately from the hardness of the ordinary  \textsc{Dominating Set}, which is W[2]-complete on split graphs and bipartite graphs \cite{Raman2008}.
In the remaining part of this section, we discuss the hardness results only for the cases $r\ge 2$ and $k\ge 2$.


\begin{theorem}\label{thm:W[2]:r}
For every fixed $k\ge 1$, \textsc{$r$-Grouped Dominating Set} is W[2]-hard when parameterized by $r$ even on split graphs.
\end{theorem}

\begin{proof}
We give a reduction from \textsc{Dominating Set} on split graphs.
Let $\langle G=(C\cup I, E),r\rangle$ be an instance  of \textsc{Dominating Set} where $C$ forms a clique and $I$ forms an independent set.  Without loss of generality, we suppose that $|C|\ge 2$ and $|I|\ge 2$.
We create $k$ copies $G_1=(V_1, E_1), \ldots, G_k=(V_k, E_k)$ of $G$ where $V_i = \{v^{(i)} \mid v\in V\}$ and $E_i = \{e^{(i)} \mid e\in E\}$. Note that $V_i = C_i \cup I_i$.
Finally, we make $\bigcup_i C_i$ a clique. 
The resulting graph $G'$ is clearly a split graph.

We show that there is a dominating set of size at most $r$ in $G$ if and only if there is an $r$-grouped dominating set with at most $k$ units in $G'$.

Suppose that there is a dominating set $D$ of size at most $r$ in $G$.
Without loss of generality, we can assume that $D\subseteq C$
\cite{Bertossi1984}.
Then we define $D'_i=\{v^{(i)}\mid v\in D\}$ for $1\le i\le k$  and $D' = \bigcup_i D'_i$.
Since $D$ is a dominating set in $G$, so is $D'_i$ on $G_i$ for each $i$. Thus, $D'$ is a dominating set in $G'$.
Because $G_i$ is a split graph and  a clique and $D'_i\subseteq  C_i$, $D'_i$ is a connected dominating set of $G_i$ of size at most $r$. 
If $|D'_i|<r$, we arbitrarily add $r-|D'_i|$ vertices in $G_i$ to $D'_i$. Then, we have a connected dominating set $D'_i$ of $G_i$ of size exactly $r$ for each $i$, which can be regarded as a unit of size $r$ of an $r$-grouped dominating set.
Clearly, $\{D'_i\mid 1\le i\le k\}$ is an $r$-grouped dominating set with $k$ units in $G'$

Conversely, suppose that there is an $r$-grouped dominating set $\mathcal{D}$ with at most $k$ units in $G'$. 
Then there is a vertex set $D_i=\bigcup \mathcal{D}\cap V_i$ of size at most $r$ in some $G_i$ by $|\bigcup \mathcal{D}|\le rk$.
Since $|I_i|\ge 2$,  $D_i$ contains at least one vertex in $C_i$. Moreover, any vertex not in $V_i$ cannot dominate vertices in $I_i$. This means that $D_i$ is a dominating set in $G_i$.
Since $G_i$ is a copy of $G$, there is a dominating set of size at most $r$ in $G$.
\end{proof}

By a similar reduction, we also show that \textsc{$r$-Grouped Dominating Set} is W[2]-hard when parameterized by $k$.

\begin{theorem}\label{thm:W[2]:k:split}
For every fixed $r\ge 1$, \textsc{$r$-Grouped Dominating Set} is W[2]-hard when parameterized by $k$ even on split graphs.
\end{theorem}

\begin{proof}
We give a reduction from \textsc{Dominating Set} on split graphs.
Let $\langle G=(C\cup I, E),k\rangle$ be an instance  of \textsc{ Dominating Set} where $C$ forms a clique and $I$ forms an independent set. Without loss of generality, we suppose that $|C|\ge 2$ and $|I|\ge 2$.
We create $r$ copies $G_1=(V_1, E_1), \ldots, G_r=(V_r, E_r)$ of $G$ where $V_i = \{v^{(i)} \mid v\in V\}$ and $E_i = \{e^{(i)} \mid e\in E\}$. Note that $V_i = C_i \cup I_i$.
Then we make $\bigcup_i C_i$ a clique. 
The resulting graph $G'$ is clearly a split graph.

We show that there is a dominating set of size at most $k$ in $G$ if and only if there is an $r$-grouped dominating set with at most $k$ units in $G'$.

Suppose that there is a dominating set $D$ of size at most $k$ in $G$.
Without loss of generality, we can assume that $D\subseteq C$.
For each $v\in D$, we define $D_v=\{v^{(i)}\mid 1\le i\le r\}$. Furthermore, let $\mathcal{D}=\{D_v \mid v\in D\}$. We see that $\mathcal{D}$ is an $r$-grouped dominating set with at most $k$ units in $G'$.
Since $D$ is a dominating set in $G$ and $G'$ consists of $r$ copies of $G$, $\bigcup \mathcal{D}$ is clearly a dominating set in $G'$.
Furthermore, because $\bigcup_i C_i$ is a clique, each $D_v$ forms a clique of size $r$, which can be regarded as a unit. 
By the assumption that $|D|\le k$, we conclude that $\mathcal{D}$ is an $r$-grouped dominating set with at most $k$ units in $G'$.

Conversely, suppose that there is an $r$-grouped dominating set $\mathcal{D}$ with at most $k$ units in $G'$. 
We see that there is a dominating set $D$ of size at most $k$ in some $G_i$.
Indeed, $\bigcup \mathcal{D}\cap V_i$ is a dominating set $D$ of size at most $k$ in  $G_i$ because
$|\bigcup \mathcal{D}|\le rk$ and there is a vertex in $\bigcup \mathcal{D}\cap C_i$ by $|I|\ge 2$. 
This completes the proof.
\end{proof}

Furthermore, we show the W[2]-hardness of \textsc{$r$-Grouped Dominating Set} on bipartite graphs.

\begin{theorem}\label{thm:W[2]:r:bipartite}
For every fixed $k\ge 1$, \textsc{$r$-Grouped Dominating Set} is W[2]-hard when parameterized by $r$ even on bipartite graphs.
\end{theorem}

\begin{proof}
We reduce \textsc{Dominating Set} on split graphs to \textsc{$r$-Grouped Dominating Set}.

We are given an instance $\langle G=C\cup I, E),r\rangle$  of \textsc{Dominating Set}. Without loss of generality,  we assume $|C|\ge 2$, $|I|\ge 2$.
Moreover, we assume that if  $\langle G=C\cup I, E),r\rangle$  is a
yes-instance, there is a dominating set $D$ of size at most $r$ such that $D\subseteq C$ \cite{Bertossi1984}.

We first delete all the edges in the clique $C$, and add two edges $\{s_1,t\},\{s_2,t\}$ and connect $t$ to all the vertices in $C$. 
The obtained graph is bipartite.
We then create $k$ copies $G_1=(V_1\cup \{s_1^{(1)},s_2^{(1)},t^{(1)}\},E_1), \ldots, G_k=(V_k\cup \{s_1^{(k)},s_2^{(k)},t^{(k)}\},E_k)$ of the graph where  $V_i=\{u^{(i)}\mid u\in V\}$ for $1\le i\le k$. To connect $G_1, \ldots, G_k$, we add edges $\{s_1^{(i)},s_1^{(i+1)}\}$ for $1\le i\le k-1$. The resulting graph denoted by $G'$ remains bipartite.

In the following, we show that there is a dominating set of size at most $r$ in $G$ if and only if there is an $(r+1)$-grouped dominating set with $k$ units in $G'$.
Suppose that there is a dominating set $D\subseteq C$ of size at most $r$ in $G$. We assume that $|D|=r$ because otherwise we only have to add $r-|D|$ vertices in $G$ to $D$ arbitrarily.
For each graph $G_i$, define $D_i = \{v^{(i)}\mid v\in D\}$. Since $t^{(i)}$ is connected to $s_1^{(i)}$, $s_2^{(i)}$, and all the vertices in the clique part $C_i$ of $G_i$ and $D_i$ is a dominating set in $G_i[V_i]$, $D_i\cup \{t^{(i)}\}$ is a connected dominating set of size $r+1$ in $G_i$. Therefore, $\{D_i\cup \{t^{(i)}\}\mid 1\le i\le k\}$ is an $(r+1)$-grouped dominating set with $k$ units in $G'$.

Conversely, let $\mathcal{D}$ be an $(r+1)$-grouped dominating set with $k$ units in $G'$.
Since $|\bigcup \mathcal{D}|\le (r+1)k$, some $G_i$ satisfies $|\bigcup \mathcal{D}\cap V_i\cup \{s_1^{(i)},s_2^{(i)},t^{(i)}\}|\le (r+1)$. To dominate $s_2^{(i)}$, $\bigcup \mathcal{D}$ must contains $t^{(i)}$. Note that $r\ge 2$. Thus, $|\bigcup \mathcal{D}\cap V_i|\le r$. 
Furthermore, $G_i$ is bipartite and $|I|\ge 2$, hence there is a vertex $v^{(i)}$ in $\bigcup \mathcal{D}\cap C_i$.
Let $D\subseteq V$ be a set in $G$ corresponding to $\bigcup \mathcal{D}\cap V_i$.
Then $D$ is a dominating set of size $r$ in $G$. Indeed, since $\bigcup \mathcal{D}\cap V_i$ dominates the independent set part of $G_i$, $D$ also dominates the independent set part of $G$.
Moreover, since $G$ is a split graph,  vertex  $v\in V$ corresponding to $v^{(i)}$ dominates all the vertices in $D$.
This completes the proof.
\end{proof}

\begin{theorem}\label{thm:W[2]:k:bipartite}
For every fixed $r\ge 1$, \textsc{$r$-Grouped Dominating Set} is W[2]-hard when parameterized by $k$ even on bipartite graphs.
\end{theorem}

\begin{proof}
We reduce \textsc{Dominating Set} on split graphs to \textsc{$r$-Grouped Dominating Set}.

We are given an instance $\langle G=C\cup I, E),k \rangle$  of \textsc{Dominating Set}. Without loss of generality, if  $\langle G=C\cup I, E),k\rangle$  is a yes-instance, there is a dominating set $D$ of size at most $k$ such that $D\subseteq C$ \cite{Bertossi1984}.
Then we construct a bipartite graph $G'$ as follows.
First, delete all the edges in the clique $C$. Then $G$ becomes a bipartite graph. We next add $k$ paths of length $r$ and connect an endpoint of each path to all the vertices in $C$. Let $P_i = (u^{(i)}_1, \ldots, u^{(i)}_r)$ denote such paths for $1\le i\le k$, and $u^{(i)}_1$'s are the endpoints connected to $C$.
The resulting graph $G'$ is bipartite.

Suppose that $G$ has a dominating set $D=\{v_1, \ldots, v_{|D|}\}\subseteq C$ of size at most $k$. From $D$, we construct an $r$-grouped dominating set with $k$ units. For each $v_i\in D$, we choose a path $v_i, u^{(i)}_1, \ldots, u^{(i)}_{r-1}$ of length $r$ as one unit of the $r$-grouped dominating set. If $|D| < k$, we choose the remaining $k-|D|$ paths $u^{(i)}_1, \ldots, u^{(i)}_{r}$ for $|D|+1\le i\le r$. Let $\mathcal{D}$ be the set of such $k$ paths. Then $\mathcal{D}$ is an $r$-grouped dominating set with $k$ units because the length of each path in $\mathcal{D}$ is $r$ and $\bigcup \mathcal{D}$ contains $D$, which dominates all the vertices in the original $G$ and the vertices in $P_i$'s.

Conversely, let $\mathcal{D}$ be an $r$-grouped dominating set with $k$ units in $G'$.
To dominate an endpoint $u^{(i)}_r$ in $P_i$, $\bigcup \mathcal{D}$ must contain $u^{(i)}_{r-1}$, which implies $|\bigcup \mathcal{D} \cap \bigcup_i P_i|\ge k(r-1)$. Thus, we have $|(C\cup I) \cap \bigcup \mathcal{D}|\le k$. Since $\bigcup \mathcal{D}$ is a dominating set of $G'$ and any vertex in $P_i$'s cannot dominate $I$, $(C\cup I) \cap \bigcup \mathcal{D}$ is a dominating set of size $k$ in $G$.
\end{proof}

On the other hand, we can show that the problem is XP  when parameterized by $k+r$.
\begin{theorem}
\textsc{$r$-Grouped Dominating Set} can be solved in  $O^*(\Delta^{O(kr^2)})$ time.
\end{theorem}
\begin{proof}
We guess the candidates of $r$-grouped dominating sets with at most $k$ units. 
We first pick an arbitrary vertex $v$ and branch $d(v)+1$ cases. One case is that $v$ is contained in $\mathcal{D}$. The vertices in the unit containing $v$ is reachable from $v$ via at most $r-1$ edges. Since the number of such vertices is at most $\Delta^{r-1}$, the choice of the other $r-1$ vertices is at most $\binom{\Delta^{r-1}}{r-1}=\Delta^{O(r^2)}$. Thus the number of candidates of units that contains $v$ is $\Delta^{O(r^2)}$. 
Another case is that $v$ is not contained in $\mathcal{D}$. Then at least one neighbor of $v$ is contained in  $\mathcal{D}$. The number of candidates of units that contain it is also $\Delta^{O(r^2)}$.
Therefore, the total number of candidates of units that dominate $v$ is $\Delta^{O(r^2)}$.
After guessing one unit, we repeatedly pick a non-dominated vertex and branch as above.
The repetition occurs at most $k$ times. Thus, the total running time is $\Delta^{O(kr^2)}$.
\end{proof}

\begin{corollary}
\textsc{$r$-Grouped Dominating Set} belongs to XP when parameterized by $k+r$.
\end{corollary}

Tripathi et al.~\cite{TripathiKPPW22} showed that \textsc{Paired Dominating Set} (equivalently, \textsc{$2$-Grouped Dominating Set}) is NP-complete for planar graphs with maximum degree~5. 
We show that \textsc{$r$-Grouped Dominating Set} is NP-hard even on planar bipartite graphs of maximum degree 3 for every fixed $r\ge 1$. This strengthens the result by Tripathi et al.~\cite{TripathiKPPW22}.

\begin{theorem}\label{thm:NP-c:planar}
For every fixed $r\ge 1$, \textsc{$r$-Grouped Dominating Set} is NP-complete on planar bipartite graphs of maximum degree 3.
\end{theorem}
\begin{proof}
We reduce \textsc{Restricted Planar 3-SAT} to \textsc{$r$-Grouped Dominating Set}.
\textsc{Restricted Planar 3-SAT} is a variant of \textsc{Planar 3-SAT} such that each variable occurs in exactly three clauses, in at most two clauses positively and in at most two clauses negatively.
It is known that \textsc{Restricted Planar 3-SAT} is NP-complete~\cite{MiddendorfP93}. 

Let $\phi$ be an instance of \textsc{Restricted Planar 3-SAT}, $n$ and $m$ be the number of variables and clauses of $\phi$, respectively.
The incidence graph of $\phi$ is a bipartite graph such that it consists of variable vertices $v_{x_i}$'s corresponding to variables and clause vertices $c_j$'s corresponding to clauses.  A variable vertex $v_{x_i}$ is connected to a clause variable $c_j$ if $C_j$ has a literal of $x_i$.
The incidence graph of $\phi$ is planar.

For the incidence graph of $\phi$, we construct the graph $G=(V,E)$ by replacing variable vertices by variable gadgets.
For each variable $x_i$, its variable gadget is constructed as follows. We create three vertices $v_{x_i}, v_{\bar{x}_i}, y_i$, and then add edges $\{v_{x_i}, y_i\}$, $\{v_{\bar{x}_i}, y_i\}$.
Furthermore, we attach a path $P^r_i = y_i z^{(1)}_i z^{(2)}_i\cdots z^{(r-1)}_i$ of length $r-1$ for each $y_i$. Here, we define $z^{(0)}_i = y_i$.
Let $V_X = \{v_{x_i}, v_{\bar{x}_i} \mid i\in \{1,\ldots,n\}\}$ and $V_C = \{c_j \mid j\in \{1,\ldots,m\}\}$.
For each variable $x_i$, $v_{x_i}$ is connected to $c_j$  if $C_j$ has a positive literal of $x_i$, and $v_{\bar{x}_i}$ is connected to $c_j$ if $C_j$ has a negative literal of $x_i$. 
We complete the construction of the graph $G=(V,E)$. Figure \ref{fig:NP-hard_planar} shows a concrete example of $G=(V,E)$ for $\phi$.
Notice that $G$ is bipartite because $V_X$ and  $V_C$ form independent sets, respectively, and $P^r_i$ is a path. Furthermore, $G$ is planar because the incidence graph of $\phi$ and the variable gadgets are planar.
Finally, since each variable occurs in exactly three clauses, in at most two clauses positively and in at most two clauses negatively, the maximum degree of $G$ is at most 3.

\begin{figure}[tbp]
    \centering
    \includegraphics[width=10cm]{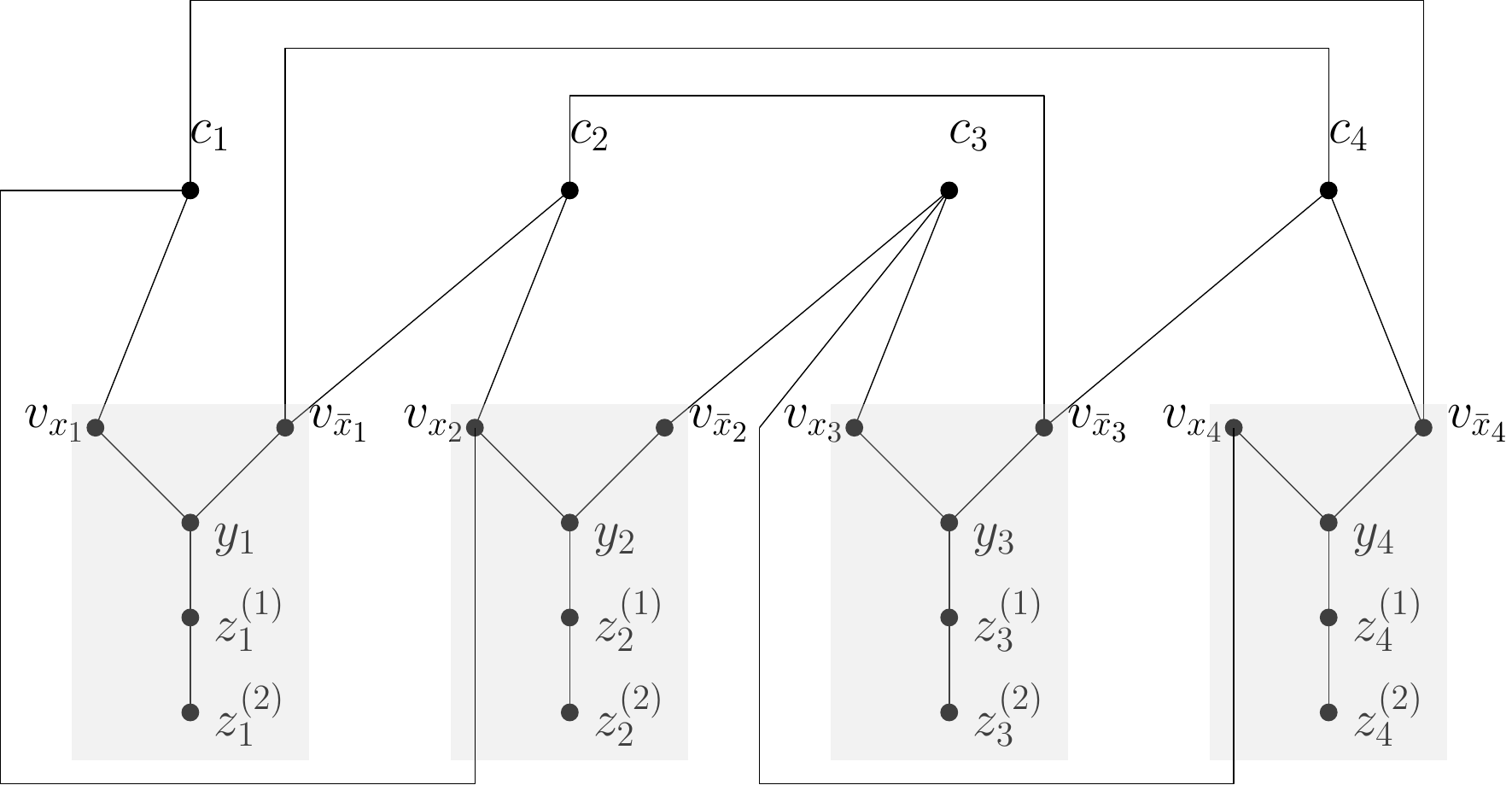}
    \caption{The graph $G$ obtained by the reduction from an instance $\phi = (x_1\lor x_2 \lor \bar{x}_4)(\bar{x}_1\lor x_2\lor \bar{x}_3)(\bar{x}_2\lor x_3\lor x_4)(\bar{x}_1\lor \bar{x}_3\lor \bar{x}_4)$ of \textsc{Restricted Planar 3-SAT} to \textsc{$3$-Grouped Dominating Set}. } 
    \label{fig:NP-hard_planar}
\end{figure}

We are ready to show that $\phi$ is a yes-instance if and only if there is an $r$-grouped dominating set with at most $n$ in $G$. 

Suppose that we are given a truth assignment of  $\phi$.
For each variable $x_i$, we select path $v_{x_i}y_i z^{(1)}_i z^{(2)}_i\cdots z^{(r-2)}_i$ as a unit of an $r$-grouped dominating set if $x_i$ is assigned to true. Otherwise, we select path $v_{\bar{x}_i}y_i z^{(1)}_i z^{(2)}_i\cdots z^{(r-2)}_i$. The number of vertices in each unit is $r$.
The unit of $x_i$ dominates vertex $z^{(r-1)}_i$. Since each clause has at least one truth literal for the truth assignment, each clause vertex $c_j$ is dominated by some unit. Therefore, the set of selected paths is an $r$-grouped dominating set with at most $n$ units.

Conversely, we are given an $r$-grouped dominating set $\mathcal{D}$ with at most $n$ in $G$. For each $i$, $z^{(r-2)}_i$ must be contained in $\bigcup \mathcal{D}$. If not, $z^{(r-1)}_i$ is not dominated because of $r\ge 2$.
Since $P^r_i$ is a path of length $r-1$ and the number of units is $n$, the vertices of a unit are selected from $\{v_{x_i},v_{\bar{x}_i}\} \cup \{y_i, z^{(1)}_i, z^{(2)}_i, \cdots, z^{(r-1)}_i\}$ for each $i$ and the unit forms a path of length $r$.
If vertex $z^{(r-1)}_i$ is contained in $\bigcup \mathcal{D}$, $\bigcup \mathcal{D}$ does not contain $v_{x_i}$ and $v_{\bar{x}_i}$. Thus, we can remove $z^{(r-1)}_i$ from $\bigcup \mathcal{D}$ and add either $v_{x_i}$ or $v_{\bar{x}_i}$ arbitrarily to $\bigcup \mathcal{D}$. Since $v_{x_i}y_i z^{(1)}_i z^{(2)}_i\cdots z^{(r-2)}_i$ is a path of length $r$, this replacement does not collapse the property of $r$-grouped dominating set.
Thus, we can suppose that the path of $i$th unit of $\mathcal{D}$ has either $v_{x_i}$ or $v_{\bar{x}_i}$ as an endpoint.
Since $\bigcup \mathcal{D}$ is a dominating set, each vertex in $V_C$ has at least one vertex in $\mathcal{D}\cap V_X$ as a neighbor.
This implies that the assignment corresponding to the selection of endpoints of units is a truth assignment.
This completes the proof.
\end{proof}

In the proof of Theorem \ref{thm:NP-c:planar}, the size of the constructed graph for $\phi$ is $O(rn+m)$. Thus, we have the following corollary.
\begin{corollary}\label{cor:ETH:hard}
For every fixed $r\ge 1$, \textsc{$r$-Grouped Dominating Set} cannot be solved in time $2^{o(n+m)}$ on bipartite graphs unless ETH fails.
\end{corollary}

\section{Fast Algorithms Parameterized by Vertex Cover Number and by Twin Cover Number}
\label{sec:fast-algorithms}

In this section, we present FPT algorithms for \textsc{$r$-Grouped Dominating Set} parameterized by vertex cover number $\nu$.  
Our algorithm is based on dynamic programming on nested partitions of a vertex cover, and its running time is $O^*((2\nu(r+1))^{\nu})$ for general $r\ge 2$.
For the cases of $r \in \{2,3\}$, we can tailor the algorithm to run in $O^*((r+1)^\nu)$ time
by focusing on the fact that the nested partitions of a vertex cover degenerate in some sense. 

We then turn our attention to a more general parameter twin cover number.
We show that, given a twin cover, \textsc{$r$-Grouped Dominating Set} admits an optimal solution 
in which twin-edges do not contribute to the connectivity of $r$-units.
This implies that these edges can be removed from the graph,
and thus we can focus on the resultant graph of bounded vertex cover number.
Hence, we can conclude that
our algorithms still work when the parameter $\nu$ in the running time is replaced with twin cover number $\tau$.

\begin{theorem}
\label{thm:r+tau_fpt}
For graphs of twin cover number $\tau$, 
\textsc{$r$-Grouped Dominating Set} can be solved in $O^*((2\tau(r+1))^{\tau})$ time. For the cases of $r \in \{2,3\}$, it can be solved in $O^*((r+1)^\tau)$ time.
\end{theorem}


With a simple observation, 
Theorem~\ref{thm:r+tau_fpt} implies that \textsc{$r$-Grouped Dominating Set} parameterized solely by $\tau$ is fixed-parameter tractable.
\begin{corollary}
\label{cor:tau^tau}
For graphs of twin cover number $\tau$, 
\textsc{$r$-Grouped Dominating Set} can be solved in $O^*((2\tau)^{2\tau})$ time.
\end{corollary}
\begin{proof}
If $r < 2\tau -1$, then the problem can be solved in $O^*((2\tau)^{2\tau})$ time by Theorem~\ref{thm:r+tau_fpt}.
Assume that $r \ge 2\tau-1$.
Let $C$ be a connected component of the input graph. If $|V(C)| < r$, then we have a trivial no-instance.
Otherwise, we construct a connected dominating set $D$ of $C$ with size exactly $r$,
which works as a unit dominating $C$.
We initialize $D$ with a non-empty twin cover of size at most $\tau$. 
Note that such a set can be found in $O^*(1.2738^{\tau})$ time:
if $C$ is a complete graph, then we pick an arbitrary vertex $v \in V(C)$ and set $D = \{v\}$;
otherwise, just find a minimum twin cover.
Since $C$ is connected, $D$ is a dominating set of $C$.
If $C[D]$ is not connected, we update $D$ with a new element $v$ adjacent to at least two connected components of $C[D]$.
Since $|D| \le \tau$ at the beginning, we can repeat this update at most $\tau -1$ times,
and after that $C[D]$ becomes connected and $|D| \le 2\tau - 1 \le r$.
We finally add $r - |D|$ vertices arbitrarily and obtain a desired set.
\end{proof}

In the next subsection, we first present an algorithm for \textsc{2-Grouped Dominating Set} parameterized by vertex cover number, which gives a basic scheme of our dynamic programming based algorithms.  We then see how we extend the idea to \textsc{3-Grouped Dominating Set}. As explained above, these algorithms are based on dynamic programming (DP), and they compute certain function values on partitions of a vertex cover.   
Unfortunately, it is not obvious how to extend the strategy to general $r$. Instead, we consider nested partitions of a vertex cover for DP tables, which makes the running time a little slower though. In the last subsection, we see how a vertex cover can be replaced with a twin cover in the same running time in terms of order. 

\subsection{Algorithms parameterized by vertex cover number}
\subsubsection{Algorithm for \textsc{2-Grouped Dominating Set}}
We first present an algorithm for the simplest case $r=2$, i.e., the paired dominating set. 
Let $G=(V,E)$ be a graph and $J$ be a vertex cover of $G$.  Then, $I=V\setminus J$ is an independent set. The basic scheme of our algorithm follows the algorithm for the dominating set problem by Liedloff~\cite{liedloff}, which focuses on a partition of a given vertex cover $J$. For a minimum dominating set $D$, the vertex cover $J$ is partitioned into three parts: $J\cap D$; $(J\setminus D)\cap N(J\cap D)$, that is, the vertices in $J\setminus D$ that are dominated by $J\cap D$; and $J\setminus N[J\cap D]$, that is, the remaining vertices. Note that the remaining vertices in $J\setminus N[J\cap D]$ are dominated by $I\cap D$. Once $J\cap D$ is fixed, a minimum $I\cap D$ is found by solving the set cover problem that reflects the condition that $J\setminus N[J\cap D]$ must be dominated by $I\cap D$. The algorithm computes a minimum dominating set by solving set cover problems defined by all candidates of $J\cap D$. 

To adjust the algorithm to \textsc{2-Grouped Dominating Set}, we need to handle the condition that a dominating set contains a perfect matching.  

For each subset $J_D\subseteq J$, we find a subset $I_D\subseteq I$ (if any exists) of the minimum size such that $J_D \cup I_D$ can form a 2-grouped dominating set. Let $X$ and $Y$ be disjoint subsets of $J$, and let $I=\{v_1,v_2,\ldots,v_{|I|}\}$ (see Fig.~\ref{fig:pd}). 
\begin{figure}[tb]
 \centering
 \includegraphics[scale=0.18]{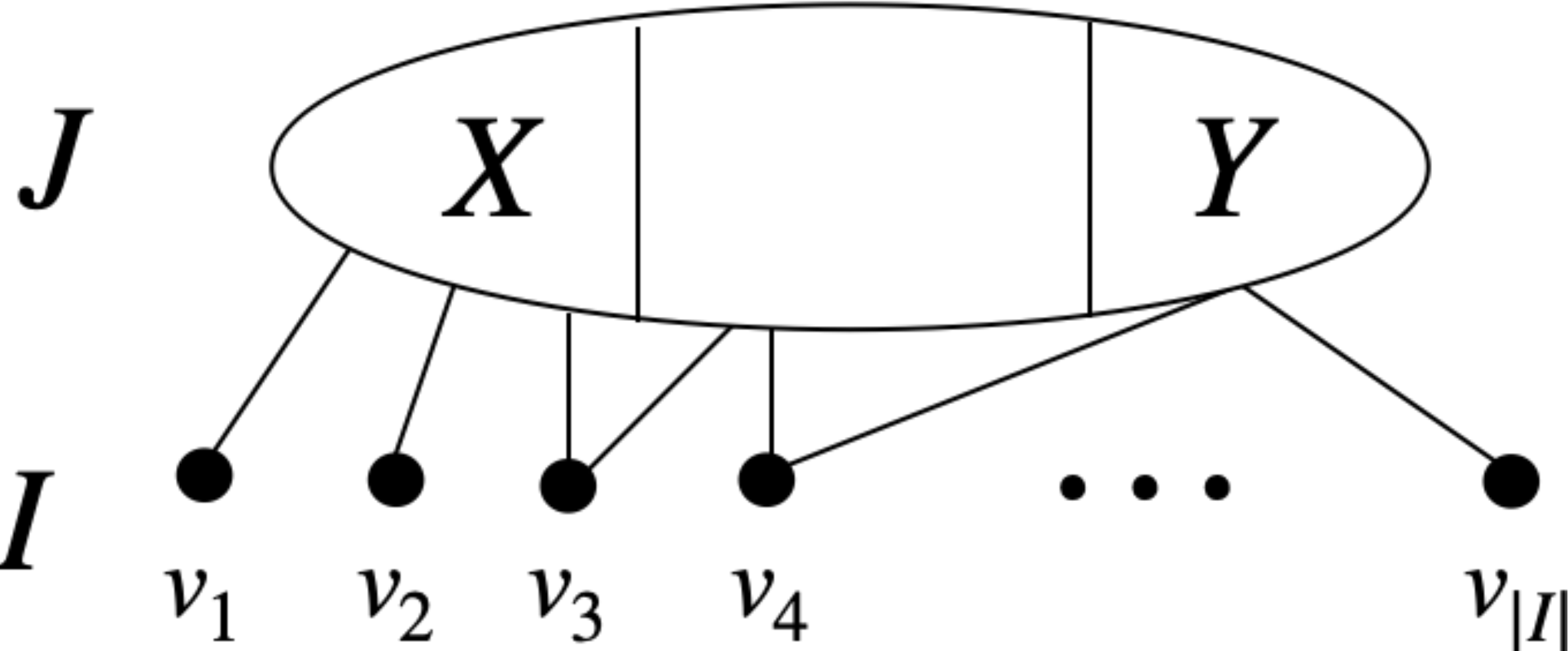}
 \caption{Partitioning a vertex cover into three parts.}
 \label{fig:pd}
\end{figure}
For $j=0, \ldots,|I|$, we define an auxiliary table $\OPT[X,Y,j]$ as the minimum size of $I'\subseteq \{v_1,v_2,\ldots,v_j\}$ that satisfies the following conditions.
\begin{enumerate}
    \item $Y\subseteq N(I')$, 
    \item $I'\cup X$ has  a partition  $\mathcal D^{(2)}=\{D^{(2)}_{1},D^{(2)}_{2},\ldots,D^{(2)}_{p}\}$ with $p\le k$ such that for all $i=1,\ldots,p$, $|D^{(2)}_{i}|=2$ and $G[D^{(2)}_{i}]$ is connected.
\end{enumerate}
We set $\OPT[X,Y,j]=\infty$ if no $I'\subseteq \{v_1,v_2,\ldots,v_j\}$ satisfies the conditions.
We can easily compute $\OPT[X,Y,0] \in \{0, \infty\}$ as
$\OPT[X,Y,0] = 0$ if and only if $G[X]$ has a perfect matching and $Y = \emptyset$.
Now the following recurrence formula computes $\OPT$:
\begin{align*}
\OPT[X,Y,j+1] & = \min \left\{
     \OPT[X,Y,j], \ 
     \underset{u\in N(v_{j+1}) \cap X}{\min} \OPT[X \setminus \{ u \},Y \setminus N(v_{j+1}),j]+1
\right\}.
\end{align*}
The recurrence finds the best way under the condition that we can use vertices from $v_1,v_2,\ldots,v_j, v_{j+1}$ in a dominating set: not using $v_{j+1}$, or pairing $v_{j+1}$ with $u \in N(v_{j+1}) \cap X$.
We can compute all entries of $\OPT$ in $O^*(3^{|J|})$ time in a DP manner as there are only $3^{|J|}$ ways
for choosing disjoint subsets $X$ and $Y$ of $J$.

Now we compute the minimum number of units in a $2$-grouped dominating set of $G$ (if any exists) by looking up some appropriate table entries of $\OPT$.
Let $\mathcal{D}$ be a 2-grouped dominating set of $G$ with $J_{D} = J \cap \bigcup \mathcal{D}$ and $I_{D} = I \cap \bigcup \mathcal{D}$.
Since $\bigcup \mathcal{D}$ is a dominating set with no isolated vertex in $G[\bigcup \mathcal{D}]$,
$J_{D}$ dominates all vertices in $I$. Let $J_Y = J \setminus N[J_D]$. 
Then the definition of $\OPT$ implies that $\OPT[J_D,J_Y,|I|] = |I_{D}|$.
Conversely, if $X \subseteq J$ dominates $I$, $Y = J \setminus N[X]$, and $\OPT[X, Y, |I|] \ne \infty$,
then there is a $2$-grouped dominating set with $(|X| + \OPT[X, Y, |I|])/2$ units.
Therefore, the minimum number of units in a $2$-grouped dominating set of $G$ is $\min\{(|X| + \OPT[X, J \setminus N[X],|I|])/2 \mid X \subseteq J \text{ and }  I \subseteq N(X)\}$, 
which can be computed in $O^{*}(2^{|J|})$ time given the table $\OPT$. 
Thus the total running time of the algorithm is $O^{*}(3^{|J|})$.

\subsubsection{Algorithm for \textsc{3-Grouped Dominating Set}}
Next, we consider the case $r=3$, i.e., the trio dominating set. Let $G=(V,E)$ be a graph, $J$ be a vertex cover of $G$, and $I=V\setminus J$. The basic idea is the same as the case $r=2$ except that we partition the vertex cover into four parts in the DP, and thus the recurrence formula for $\OPT$ is different.
In the DP, the vertex cover $J$ is partitioned into four parts depending on the partial solution corresponding to each table entry.

For each subset $J_D\subseteq J$, we find a subset $I_D\subseteq I$ (if any exists) of the minimum size such that $J_D \cup I_D$ can form a 3-grouped dominating set.
Intuitively, the set $F$ represents partial units that will later be completed to full units.
Let $X$, $F$, and $Y$ be disjoint subsets of $J$, and let $I=\{v_1,v_2,\ldots,v_{|I|}\}$.  
For $j=0,\ldots,|I|$, we define $\OPT[X,F,Y,j]$ as the minimum size of $I'\subseteq \{v_1,v_2,\ldots,v_j\}$ that satisfies the following conditions: 
\begin{enumerate}
    \item $Y\subseteq N(I')$, 
    \item $I'$ can be partitioned into two parts $I'_2,I'_3$ satisfying the following conditions: 
    \begin{itemize}
        \item $I'_2\cup F$ has a partition $\mathcal D^{(2)}=\{D^{(2)}_1,D^{(2)}_2,\ldots,D^{(2)}_p\}$ with $p\le k$ such that for all $i=1,\ldots,p$, $|D^{(2)}_{i}|=2$ and $G[D^{(2)}_{i}]$ is connected.
        \item $I'_3\cup X$ has a partition $\mathcal D^{(3)}=\{D^{(3)}_1,D^{(3)}_2,\ldots,D^{(3)}_q\}$ with $q\le k$ such that for all $i=1,\ldots,q$, $|D^{(3)}_{i}|=3$ and $G[D^{(3)}_{i}]$ is connected.
    \end{itemize}
\end{enumerate}
We set $\OPT[X,F,Y,j]=\infty$ if no $I'\subseteq \{v_1,v_2,\ldots,v_j\}$ satisfies the conditions.
We can easily compute $\OPT[X,F,Y,0] \in \{0, \infty\}$ as
$\OPT[X,F,Y,0] = 0$ if and only if $F = Y = \emptyset$ and 
$G[X]$ admits a partition into connected graphs of $3$-vertices.
The last condition can be checked in $O(2^{|J|} \cdot |J|^{3})$ time for all $X \subseteq J$
by recursively considering all possible ways for removing three vertices from $X$;
that is, $\OPT[X,\emptyset,\emptyset,0] = \min_{\{x,y,z\} \in \binom{X}{3}} A[X \setminus \{x,y,z\}, \emptyset, \emptyset, 0]$ if $|X| \ge 3$.
The following recurrence formula holds: 
$\OPT[X,F,Y,j+1] = \min \{f_1, f_2, f_3, f_4\}$, where
\begin{align*}
  f_1 &= \OPT[X,F,Y,j],
  \\
  f_2 &= \min_{\alpha,\beta\in X,|E(G[\{\alpha,\beta,v_{j+1}\}])|\geq 2} \OPT[X\setminus\{\alpha,\beta\},F, Y\setminus \{N(\{\alpha,\beta,v_{j+1}\})\},j]+1, 
  \\
  f_3 &= \min_{\alpha\in X\cap N(v_{j+1}) } \OPT[X\setminus\{\alpha\},F\cup \{\alpha\},Y\setminus N(v_{j+1}),j]+1,
  \\
  f_4 &= \min_{\beta\in F\cap N(v_{j+1})} \OPT[X,F\setminus \{\beta\},Y\setminus N(\{\beta,v_{j+1}\}),j]+1.
\end{align*}
The four options $f_{1},f_{2},f_{3}$, and $f_{4}$ assume different ways of the role of $v_{j+1}$
and compute the optimal value under the assumptions (see Fig.~\ref{fig:tdre}). 
Concretely, $f_{1}$ reflects the case when $v_{j+1}$ does not belong to the solution, and 
$f_{2}$ reflects the case when $v_{j+1}$ belongs to the solution together with two vertices in $J$ in a connected way. In $f_{3}$, it reflects that $v_{j+1}$ forms a triple in the solution with a vertex in $F$ and a vertex in $I_j$. In $f_{4}$, it reflects that $v_{j+1}$ currently forms a pair in $J$ and will form a triple with a vertex in $I\setminus I_j$.  
We can compute all entries of $\OPT$ in $O^*(4^{|J|})$ time as the number of combinations of three disjoint sets $X,F,Y$ of $J$ is $4^{|J|}$.

\begin{figure}[tb]
 \centering
 \includegraphics[scale=0.16]{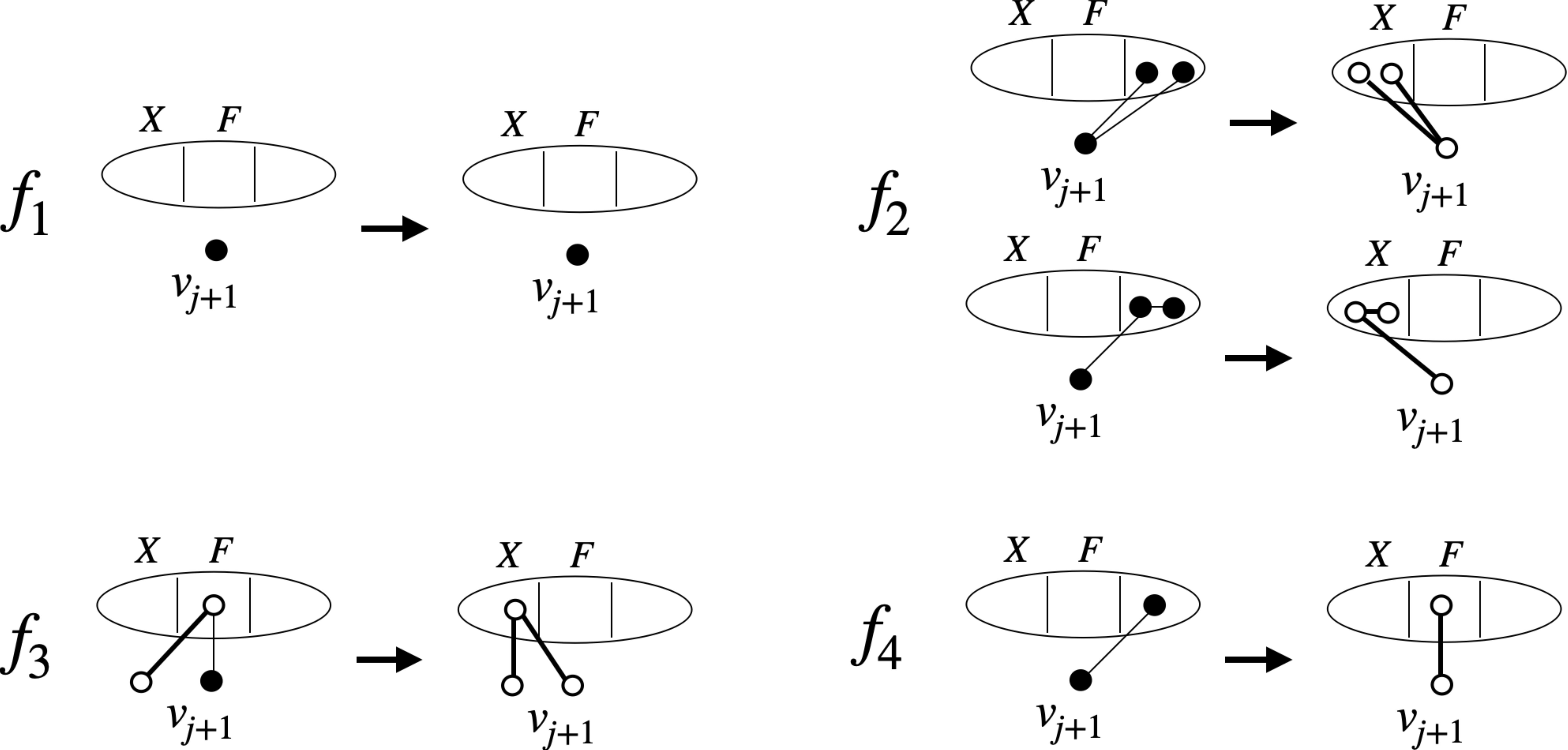}
 \caption{How $v_{j+1}$ is used. (The white vertices belong to a dominating set.)}
 \label{fig:tdre}
\end{figure}

Similarly to the previous case of $r = 2$, we can compute the minimum number of units
in a $3$-grouped dominating set as $\min\{(|X| + \OPT[X, \emptyset, J \setminus N[X],|I|])/3 \mid X \subseteq J \text{ and }  I \subseteq N(X)\}$. 
Given the table $\OPT$, this can be done in $O^{*}(2^{|J|})$ time.
Thus the total running time of the algorithm is $O^{*}(4^{|J|})$.

\subsubsection{Algorithm for \textsc{$r$-Grouped Dominating Set}}
We now present our algorithm for general $r\ge 4$. Let $G=(V,E)$ be a graph, $J$ be a vertex cover of $G$, and $I=V\setminus J$.
This case still allows an algorithm based on a similar framework to the previous cases, though connected components of general $r$ can be built up from smaller fragments of connected components; this yields an essential difference that worsens the running time.
In the DP, the vertex cover $J$ is partitioned into $r+1$ parts depending on the partial solution corresponding to each table entry,
and then some of the parts in the partition are further partitioned into smaller subsets.
In other words, each table entry corresponds to a nested partition of the vertex cover.

As in the previous algorithms, for each subset $J_D\subseteq J$, we find a subset $I_D \subseteq I$ (if any exists) of the minimum size such that $J_D \cup I_D$ can form an $r$-grouped dominating set.
Let $X$, $F^{(r-1)},\dots,F^{(3)},F^{(2)},Y$ be disjoint subsets of $J$, and let $I=\{v_1,v_2,\ldots,v_{|I|}\}$. 
For $i=2,\ldots,r-1$, let $\mathcal{F}^{(i)}$ be a partition of $F^{(i)}$, where $\mathcal{F}^{(i)}=\{F^{(i)}_1,F^{(i)}_2,\ldots,F^{(i)}_{|\mathcal{F}^{(i)}|}\}$.
The number of such nested partitions $(X,\mathcal{F}^{(r-1)},\dots,\mathcal{F}^{(2)},Y)$ is at most $(r+1)^{|J|}|J|^{|J|}$.
\begin{figure}[tb]
 \centering
 \includegraphics[scale=0.20]{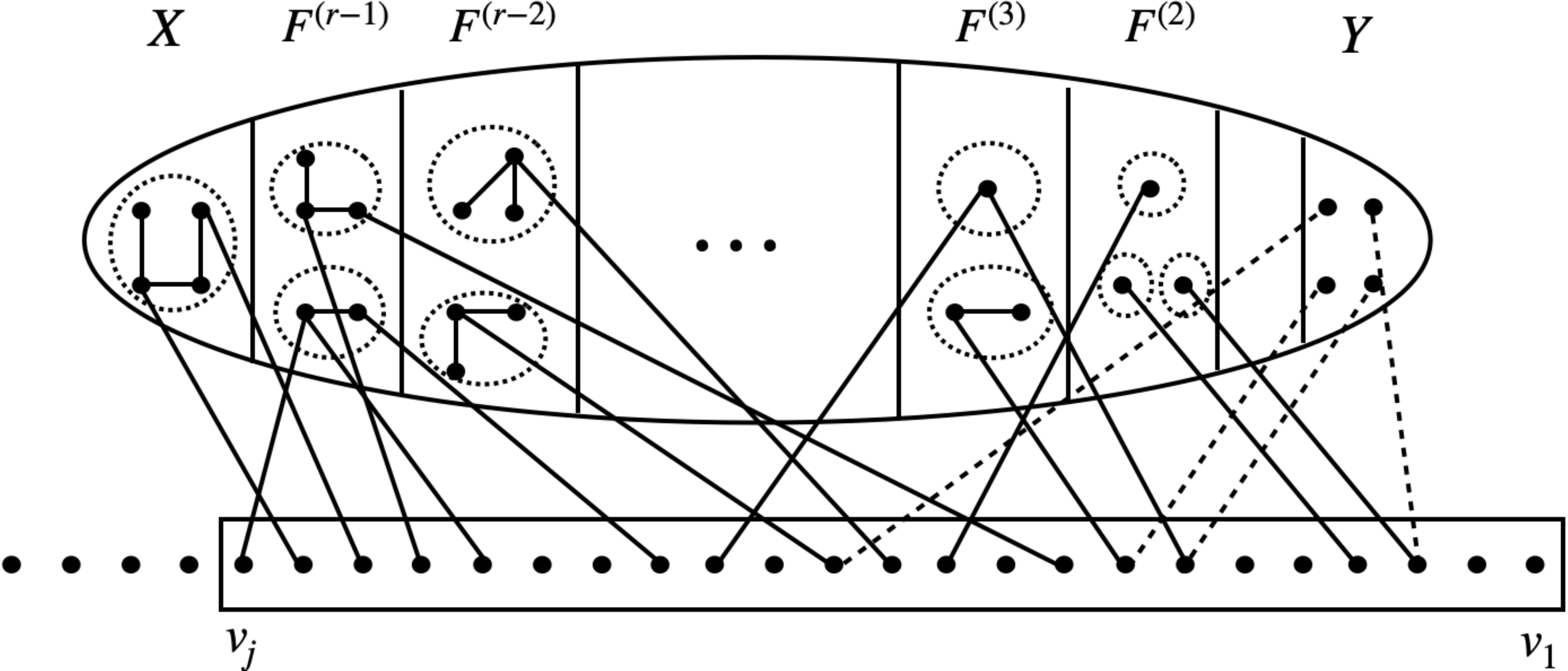}
 \caption{A nested partition of a vertex cover.}
 \label{fig:rd}
\end{figure}
For $j=0,\ldots,|I|$, we define $\OPT[X,\mathcal{F}^{(r-1)},\dots,\mathcal{F}^{(2)},Y,j]$ as the minimum size of $I'\subseteq \{v_1,v_2,\ldots,v_j\}$ that satisfies the following conditions: 
\begin{enumerate}
    \item $Y\subseteq N(I')$, 
    \item $I'$ can be partitioned into $r-1$ parts $I'_2,I'_3,\dots,I'_r$ satisfying the following conditions: 
    \begin{itemize}
        \item for $i=2,\ldots,r-1$, $I'_i \cup F^{(i)}$ has a partition $\mathcal D^{(i)}=\{D^{(i)}_1,D^{(i)}_2,\ldots,D^{(i)}_{|\mathcal{F}^{(i)}|}\}$ such that 
    for all $p=1,\ldots, |\mathcal{F}^{(i)}|$, $D^{(i)}_p$ includes at least one vertex of $I'$ and is a superset of ${F}^{(i)}_p$, and $|D^{(i)}_{p}|=i$ and $G[D^{(i)}_{p}]$ is connected.
        \item $I'_r\cup X$ has a partition $\mathcal D^{(r)}=\{D^{(r)}_1,D^{(r)}_2,\ldots,D^{(r)}_{q}\}$ such that for all $i=1,\ldots,q$, $|D^{(r)}_{i}|=r$ and $G[D^{(r)}_{i}]$ is connected.
    \end{itemize}
\end{enumerate}
We set $\OPT[X,\mathcal{F}^{(r-1)},\dots,\mathcal{F}^{(2)},Y,j] = \infty$ if no $I'\subseteq \{v_1,v_2,\ldots,v_j\}$ satisfies the conditions.
We can compute $\OPT[X,\mathcal{F}^{(r-1)},\dots,\mathcal{F}^{(2)},Y,0]$, which is $0$ or $\infty$,
as it is $0$ if and only if ${F}^{(r-1)} = \dots = {F}^{(2)} = Y = \emptyset$ and 
$G[X]$ admits a partition into connected graphs of $r$ vertices.
The last condition can be checked in $O(|J|^{|J|})$ time for all $X \subseteq J$
by checking all possible partitions of $J$.

Assume that all entries of $\OPT$ with $j \le c$ for some $c$ are computed. Since the degree of $v_{c+1}$ is at most $|J|$,
the number of possible ways of how $v_{c+1}$ extends a partial solution is at most $2^{|J|}$.
Thus from each table entry of $\OPT$ with $j = c$,
we obtain at most $2^{|J|}$ candidates of the table entries with $j = c+1$.
Thus, we can compute all entries of $\OPT$ in $O^{*}(2^{|J|}(r+1)^{|J|}|J|^{|J|})$ time.

Given $\OPT$, we can compute the minimum number of units in an $r$-grouped dominating set as $\min\{(|X| + \OPT[X, \emptyset, \dots, \emptyset, J \setminus N[X],|I|])/r \mid X \subseteq J \text{ and }  I \subseteq N(X)\}$.
Again this takes only $O^{*}(2^{|J|})$ time.
Thus the total running time of the algorithm is $O^{*}(2^{|J|}(r+1)^{|J|}|J|^{|J|})=O^{*}((2\nu(r+1))^{\nu})$.

\subsection{Algorithms parameterized by twin cover number}
In this subsection, we show that the algorithms presented above still work when the parameter $\nu$ in the running time is replaced with twin cover number $\tau$. To show this, we prove the following lemma. It says that twin-edges do not contribute to the connectivity of units for some minimum $r$-grouped dominating sets and can be removed from the graph. As a result, the vertex cover number can be replaced with the twin cover number. 
\begin{lemma}\label{lem:twin-K}
Let $G$ be a graph and $K$ be a twin cover of $G$.
If $G$ has an $r$-grouped dominating set, 
then there exists a minimum $r$-grouped dominating set such that every unit has at least one vertex in $K$. 
\end{lemma}
\begin{proof}
Let $G=(V,E)$ be a graph, and $K$ be a twin cover of $G$. Suppose that a minimum $r$-grouped dominating set $\mathcal D$ exists and one of its units $D=\{v_1,v_2,\ldots,v_r\}$ has no vertex in $K$. Since $K$ is a twin cover, $N[v_1]=N[v_2]=\dots=N[v_r]$ holds. Let $K_{D} = K\cap N(v_1)$. Then, there is at least one vertex $x$ in $K_{D}$ such that $x \notin \bigcup \mathcal D$. Suppose to the contrary that there is no such $x$, and thus all vertices in $K_{D}$ belong to $\bigcup \mathcal D$. This implies that no vertex is dominated only by $D$ and that $D$ itself is dominated by some vertices in $K_{D}$. Thus, $\mathcal{D} \setminus \{D\}$ is an $r$-grouped dominating set. 
This contradicts the minimality of $\mathcal{D}$.
Let $D' = D\setminus\{v_1\}\cup \{x\}$, then $\mathcal D'=\mathcal D\setminus \{D\}\cup \{D'\}$ is also a minimum $r$-grouped dominating set of $G$ (see Fig. \ref{fig:tclem}). By repeating this process, we can obtain a minimum $r$-grouped dominating set such that every unit has at least one vertex in $K$.
\end{proof}
\begin{figure}[h]
 \centering
 \includegraphics[scale=0.15]{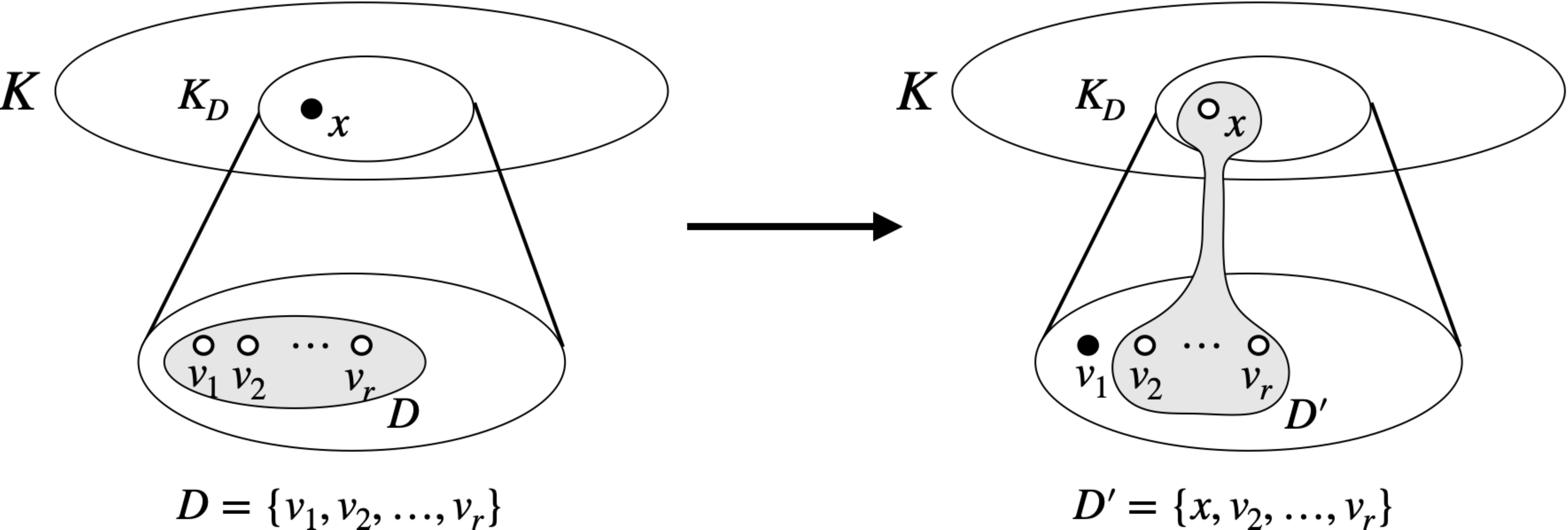}
 \caption{An example for exchange. White vertices belong to a dominating set. }
 \label{fig:tclem}
\end{figure}

\input{beyond}

%
%
%
\bibliographystyle{splncs04}
\bibliography{ref_with_doi}
\end{document}

%% file: parameters.tex


\definecolor{lightlightgray}{rgb}{0.85,0.85,0.85}

\begin{tikzpicture}[every node/.style={draw, thick, rectangle, rounded corners,align=center},scale=0.9]
  \scriptsize

  \draw[ultra thick,blue]  (-4.5,5.5) -- (4.5,5.5);
  \draw[ultra thick,blue,->] (-4.4,5.5) -- (-4.4,5);
  \node[draw=none,blue] at (-2.75,5.1) {FPT (${}+{r+k}$) \\ Corollaries~\ref{cor:r+k_nwd} and \ref{cor:r+k+twin-width}};

  \draw[ultra thick,blue]  (2,4.1) -- (4.5,4.1);
  \draw[ultra thick,blue,->] (2.75,4.1) -- (2.75,3.6);
  \node[draw=none,blue] at (4,3.7) {FPT ($+r$) \\ Corollary~\ref{cor:r+tw}};

  \draw[ultra thick,darkred] (-4.5,3.6) -- (-2,3.6) to[out=0,in=90] (-1.5,3.1) -- (-1.5,1.1) to[out=270,in=180] (-1,.6) -- (2, 0.6) -- (4.5,.6);
  \draw[ultra thick,darkred,->] (-4.4,3.6) -- (-4.4,4.1);
  \node[draw=none,darkred] at (-3,4) {W[1]-hard ($+k$) \\ Theorem~\ref{thm:k+td+fvs}};

  \draw[ultra thick,blue] (-4.5,3.5) -- (-2.1,3.5) to[out=0,in=90] (-1.6,3) -- (-1.6,1) to[out=270,in=180] (-1.1,.5) -- (2, 0.5) -- (4.5,0.5);
  \draw[ultra thick,blue,->] (2.9,0.5) -- (2.9,0);
  \node[draw=none,blue] at (4,0.1) {FPT \\ Theorem~\ref{thm:mw}};
  \node[draw=none,blue] at (-3,.35) {$2^{O(\tau \log \tau)}$ time \\ Corollary~\ref{cor:tau^tau}};

  \node (twin) at  (0, 5) {twin-width};
  \node (nwd)  at  (3, 5) {nowhere dense};
  \node (cw)   at  (0, 4) {clique-width};
  \node (mw)   at (-3, 3) {modular-width};
  \node (tw)   at  (3, 3) {treewidth};
  \node (fvs)  at  (0, 2) {feedback vertex set};
  \node (pw)   at  (3, 2) {pathwidth};
  \node (tc)   at (-3, 1) {twin cover};
  \node (td)   at  (3, 1) {treedepth};
  \node (vc)   at  (0, 0) {vertex cover};

  \draw[thick] (twin)--(cw)--(mw)--(tc)--(vc);
  \draw[thick] (nwd)--(tw);
  \draw[thick] (cw)--(tw)--(pw)--(td)--(vc);
  \draw[thick] (tw)--(fvs)--(vc);
\end{tikzpicture}

%% file: beyond.tex


\section{Beyond Vertex Cover and Twin Cover}

In this section,
we further explore the parameterized complexity of $r$-\textsc{Grouped Dominating Set}
with respect to structural graph parameters that generalize vertex cover number and twin cover number.
We show that if we do not try to optimize the running time of algorithms,
then we can use known algorithmic meta-theorems that 
automatically give fixed-parameter algorithms
parameterized by certain graph parameters.

For the sake of brevity, we define only the parameters for which we need their definitions.
For example, we do not need the definition of treewidth for applying the meta-theorem described below.
On the other hand, to contrast the results here with the ones in the previous sections,
it is important to see the picture of the relationships between the parameters.
See \figurename~\ref{fig:parameters0} for the hierarchy of the graph parameters we deal with.

Roughly speaking, the algorithmic meta-theorems we use here say that 
if a problem can be expressed in a certain logic (e.g., \fo{}, \mso{1}, or \mso{2}),
then the problem is fixed-parameter tractable parameterized by a certain graph parameter (e.g., twin-width, treewidth, or clique-width).
Such theorems are extremely powerful and used widely for designing fixed-parameter algorithms~\cite{Kreutzer11}.
On the other hand, the generality of the meta-theorems unfortunately comes with very high dependency on the parameters~\cite{FrickG04}.
When our target parameter is vertex cover number, the situation is slightly better,
but still a double-exponential $2^{2^{\mathrm{\Omega}(\nu)}}$ lower bound of the parameter dependency is known under ETH~\cite{Lampis12}.
This implies that our ``slightly superexponential'' $2^{O(\tau \log \tau)}$ algorithm in Section~\ref{sec:fast-algorithms}
cannot be obtained by applications of known meta-theorems.

In the rest of this section, we first introduce \fo{}, \mso{1}, and \mso{2} on graphs.
We then observe that the problem can be expressed in \fo{} when $r$ and $k$ are part of the parameter
and in \mso{2} when $r$ is part of the parameter.
These observations combined with known meta-theorems immediately imply that
$r$-\textsc{Grouped Dominating Set} is fixed-parameter tractable when
\begin{itemize}
  \item parameterized by $r+ k$ on nowhere dense graph classes;

  \item parameterized by $r + k + \textrm{twin-width}$
  if a contraction sequence of the minimum width is given as part of the input; and

  \item parameterized by $r + \text{treewidth}$.
\end{itemize}

We then consider the parameter $k + \text{treewidth}$
and show that this case is intractable. More strongly, we show that 
$r$-\textsc{Grouped Dominating Set} is W[1]-hard
when the parameter is $k + \text{treedepth} + \text{feedback vertex set number}$.

We finally consider the parameter modular-width, a generalization of twin cover number,
and show that $r$-\textsc{Grouped Dominating Set} parameterized by modular-width is fixed-parameter tractable.

\subsection{Results based on algorithmic meta-theorems}
\label{sec:logic}
The \emph{first-order logic} on graphs (\fo) allows variables representing vertices of the graph under consideration.
The atomic formulas are the equality $x=y$ of variables and the adjacency $E(x,y)$ meaning that $\{x,y\} \in E$.
The \fo{} formulas are defined recursively from atomic formulas
with the usual Boolean connectives ($\lnot$, $\land$, $\lor$, $\Rightarrow$, $\Leftrightarrow$),
and quantification of variables ($\forall$, $\exists$).
We also use the existential quantifier with a dot ($\dot{\exists}$) to quantify distinct objects.
For example, $\dot{\exists} a,b \colon \phi$ means $\exists a,b \colon (a \ne b) \land \phi$.
We write $G \models \phi$ if $G$ satisfies (or \emph{models}) $\phi$.
Given a graph $G$ and an \fo{} formula $\phi$, 
\textsc{\fo{} Model Checking} asks whether $G \models \phi$.

It is straightforward to express the property of having an $r$-grouped dominating set of $k$ units 
with an \fo{} formula whose length depends only on $r+k$:
\begin{align*}
  \phi_{r,k}
  &=
  \dot{\exists} v_{1}, v_{2}, \dots, v_{rk} \colon   \\
  &\qquad
  \mathsf{dominating}(v_{1}, \dots, v_{rk}) \land {}
  \bigwedge_{0 \le i \le k-1} 
  \mathsf{connected}(v_{ir+1}, \dots, v_{ir+r}),
\end{align*}
where $\mathsf{dominating}(\cdot\cdot\cdot)$ is a subformula expressing that the $rk$ vertices form a dominating set
and $\mathsf{connected}(\cdot\cdot\cdot)$ is the one expressing that the $r$ vertices induce a connected subgraph
(see Section~\ref{sec:subformulas} for the expressions of the subformulas).
This implies that \textsc{$r$-Grouped Dominating Set} parameterized by $r + k$
is fixed-parameter tractable on graph classes on which 
\textsc{\fo{} Model Checking} parameterized by the formula length $|\phi|$ is fixed-parameter tractable.
Such graph classes include
nowhere dense graph classes~\cite{GroheKS17} and 
graphs of bounded twin-width (given with so called contraction sequences)~\cite{BonnetKTW22}.
\begin{corollary}
\label{cor:r+k_nwd}
$r$-\textsc{Grouped Dominating Set} parameterized by $r+k$ is fixed-parameter tractable
on nowhere dense graph classes.
\end{corollary}
\begin{corollary}
\label{cor:r+k+twin-width}
$r$-\textsc{Grouped Dominating Set} parameterized by $r + k + \textrm{twin-width}$ is fixed-parameter tractable
if a contraction sequence of the minimum width is given as part of the input.
\end{corollary}

The \emph{monadic second-order logic} on graphs (\mso{1})
is an extension of \fo{} that additionally allows variables representing vertex sets
and the inclusion predicate $X(x)$ meaning that $x \in X$.
\mso{2} is a further extension of \mso{1} that also allows edge variables, edge-set variables,
and an atomic formula $I(e,x)$ representing the edge-vertex incidence relation.
Given a graph $G$ and an \mso{1} (\mso{2}, resp.) formula $\phi(X)$ with a free set variable $X$, 
\mso{1} (\mso{2}, resp.) \textsc{Optimization} asks to find a minimum set $S$ such that $G \models \phi(S)$.

It is not difficult to express the property of a vertex set being the union of $r$-units of a $r$-grouped dominating set
with an \mso{2} formula whose length depending only on $r$:\footnote{%
Note that there is no equivalent \mso{1} formula of length depending only on $r$.
This is because $G \models \psi_{2}(V)$ expresses the property of having a perfect matching,
for which an \mso{1} formula does not exist (see e.g., \cite{CourcelleE12}).}
\begin{align*}
  \psi_{r}(X)
  ={}&
  \mathsf{dominating}(X) \land {}
  \\
  &\left(\exists F \subseteq E \colon \mathsf{span}(F,X) \land
  (\forall C \subseteq X\colon
  \mathsf{cc}(F,C) \Rightarrow \mathsf{size}_{r}(C))\right),
\end{align*}
where $\mathsf{dominating}(X)$ is a subformula expressing that $X$ is a dominating set,
$\mathsf{span}(F,X)$ is the one expressing that $X$ is the set of all endpoints of the edges in $F$,
$\mathsf{cc}(F,C)$ expresses that $C$ is the vertex set of a connected component of the subgraph induced by $F$, and
$\mathsf{size}_{r}(C)$ means that $C$ contains exactly $r$ elements
(again, see Section~\ref{sec:subformulas} for the expressions of the subformulas).
Since \textsc{\mso{2} Optimization} parameterized by treewidth
is fixed-parameter tractable~\cite{ArnborgLS91,BoriePT92,Courcelle90mso1},
we have the following result.
\begin{corollary}
\label{cor:r+tw}
$r$-\textsc{Grouped Dominating Set} parameterized by $r + \text{treewidth}$ is fixed-parameter tractable.
\end{corollary}

\subsection{Hardness parameterized by $k + \text{treewidth}$}
Now the natural question regarding treewidth and $r$-\textsc{Grouped Dominating Set}
would be the complexity parameterized by $k + \text{treewidth}$.
Unfortunately, this case is W[1]-hard even if treewidth
is replaced with a possibly much larger parameter $\text{pathwidth} + \text{feedback vertex set number}$
and the graphs are restricted to be planar.
Furthermore, if the planarity is not required, 
we can replace pathwidth in the parameter with treedepth.
\begin{theorem}
$r$-\textsc{Grouped Dominating Set} parameterized by 
$k + \text{pathwidth} + \text{feedback vertex set number}$ is W[1]-hard on planar graphs.
\end{theorem}
\begin{proof}
Given a graph $G = (V,E)$ and an integer $r \ge 2$,
\textsc{Equitable Connected Partition} asks whether there exists
a partition of $V$ into $k = |V|/r$ sets $V_{1}, \dots, V_{k}$
such that $G[V_{i}]$ is connected and $|V_{i}| = r$ for $1 \le i \le k$.
It is known that \textsc{Equitable Connected Partition}
parameterized by $k + \text{pathwidth} + \text{feedback vertex set number}$
is W[1]-hard even on planar graphs~\cite{EncisoFGKRS09}.
We reduce this problem to ours.

Let $\langle G = (V,E), r \rangle$ be an instance of \textsc{Equitable Connected Partition}.
To each vertex $v$ of $G$, we attach a new vertex of degree~$1$, which we call a \emph{pendant} at $v$.
This modification does not change the feedback vertex number
and may increase the pathwidth by at most $1$ (see e.g., \cite[Lemma~A.2]{BelmonteHKKKKLO22}).
Let $H$ be the resultant graph, which is planar.
To prove the lemma, it suffices to show that 
$\langle H, k \rangle$ is a yes-instance of $r$-\textsc{Grouped Dominating Set}
if and only if
$\langle G, r \rangle$ is a yes-instance of \textsc{Equitable Connected Partition}.

To prove the if direction, assume that
$\langle G, r \rangle$ is a yes-instance of \textsc{Equitable Connected Partition}
and that $V_{1}, \dots, V_{k}$ certificate it.
Clearly, $\{V_{1}, \dots, V_{k}\}$ is an $r$-grouped dominating set of $H$.

To prove the only-if direction, assume that
$H$ has an $r$-grouped dominating set $\mathcal{D}$ with at most $k$ units.
Let $v \in V$ and $p$ be the pendant at $v$.
Observe that $\bigcup \mathcal{D}$ contains exactly one of $v$ and $p$
since it needs at least one of them for dominating $p$
and $|\bigcup \mathcal{D}| \le rk = |V|$.
Furthermore, the assumption $r \ge 2$ implies that
$\bigcup \mathcal{D}$ cannot contain $p$ as it has no neighbor other than $v$.
This implies that $\bigcup \mathcal{D} = V$ and that $\mathcal{D}$ contains exactly $|V|/r = k$ units.
Therefore, the family $\mathcal{D}$ is a certificate that $\langle G, r \rangle$
is a yes-instance of \textsc{Equitable Connected Partition}.
\end{proof}

It is known that on general (not necessarily planar) graphs,
\textsc{Equitable Connected Partition}
parameterized by $k + \text{treedepth} + \text{feedback vertex set number}$
is W[1]-hard~\cite{GimaO22_arxiv}.
Since adding pendants to all vertices increases treedepth by at most $1$ (see e.g., \cite{NesetrilO2012}),
the same reduction shows the following hardness.
\begin{theorem}
\label{thm:k+td+fvs}
$r$-\textsc{Grouped Dominating Set} parameterized by 
$k + \text{treedepth} + \text{feedback vertex set number}$ is W[1]-hard.
\end{theorem}

\subsection{Fixed-parameter tractability parameterized by modular-width}

Let $G = (V,E)$ be a graph. A set $M \subseteq V$ is a \emph{module} if
for each $v \in V \setminus M$, either $M \subseteq N(v)$ or $M \cap N(v) = \emptyset$ holds.
The \emph{modular-width} of $G$, denoted $\mw(G)$, is the minimum integer $k$
such that either $|V| \le k$ or
there exists a partition of $V$ into at most $k$ modules $M_{1}, \dots, M_{k'}$ of $G$
such that each $G[M_{i}]$ has modular-width at most $k$.
It is known that the modular-width of a graph
and a recursive partition certificating it can be computed in linear time~\cite{CournierH94,HabibP10,TedderCHP08}.

Observe that if $V$ is partitioned into modules $M_{1}, \dots, M_{k}$ of $G$,
then for two distinct modules $M_{i}$ and $M_{j}$, we have either no or all possible edges between them.
If there are all possible edges between $M_{i}$ and $M_{j}$,
then we say that $M_{i}$ and $M_{j}$ are adjacent.

\begin{theorem}
\label{thm:mw}
$r$-\textsc{Grouped Dominating Set} parameterized by modular-width is fixed-parameter tractable.
\end{theorem}
\begin{proof}
Let $\langle G = (V,E), k \rangle$ be an instance of $r$-\textsc{Grouped Dominating Set}.
We only consider the case of $r \ge 2$ since the other case of $r=1$ is known (see \cite{CourcelleMR00,GajarskyLO13}).
We may assume that $G$ is connected since otherwise we can solve the problem on each connected component separately.
We also assume that $G$ has at least $r$ vertices as otherwise the problem is trivial.
Let $M_{1}, \dots, M_{\mu}$ be a partition of $V$ into modules with $2 \le \mu \le \mw(G)$.
For each module $M_{i}$, there is at least one adjacent module $M_{j}$
as $G$ is connected.

We first assume that $r \ge \mu$.
Let $D \subseteq V$ be an arbitrary set of size $r$
that takes at least one vertex from each module $M_{i}$.
Recall that we have either no or all possible edges between two distinct modules.
Thus, the connectivity of $G$ implies that $G[D]$ is connected
and $D$ is a dominating set of $G$.
This implies that $\{D\}$ is an $r$-grouped dominating set with one unit.

Next assume that $r < \mu$.
In this case, we show below that if $G$ has an $r$-grouped dominating set,
then $G$ has an $r$-grouped dominating set with at most $\mu$ units.
This implies that $r+k < 2\mu$, and thus the problem can be solved as \textsc{\fo{} Model Checking}
with a formula of length depending only on $\mu$, which is fixed-parameter tractable parameterized by $\mu$
(see~\cite{CourcelleMR00,GajarskyLO13}).

Before proving the upper bound of $k$, we show that if $G$ has an $r$-grouped dominating set,
then there is a minimum one such that no unit is entirely contained in a module $M_{i}$.
Assume that $\mathcal{D}$ is a minimum $r$-grouped dominating set of $G$.
If $D \subseteq M_{i}$ holds for some $i$ and $D \in \mathcal{D}$,
then there is a vertex $v$ in a module $M_{j}$ adjacent to $M_{i}$ that does not belong to $\bigcup \mathcal{D}$.
This is because, otherwise, $\mathcal{D} \setminus \{D\}$ is still an $r$-grouped dominating set.
Let $u$ be an arbitrary vertex in $D$
and set $D' = D \setminus \{u\} \cup \{v\}$.
As $r \ge 2$, $D'$ intersects both $M_{i}$ and $M_{j}$.
Also we can see that $|D'| = r$, 
$D'$ is connected (as $u$ is adjacent to all vertices in $M_{i}$), and
all vertices dominated by $D$ are dominated by $D'$ as well.
Thus, $\mathcal{D} \setminus \{D\} \cup \{D'\}$ is an $r$-grouped dominating set.
We can repeat this process until we have the claimed property.

As discussed above, it suffices to show the upper bound $k$ for the number of units.
Let $\mathcal{D}$ be a minimum $r$-grouped dominating set of $G$
such that no unit is entirely contained in a module $M_{i}$.
We say that a module $M_{i}$ is \emph{private} for a unit $D \in \mathcal{D}$
if $D$ is the only one in $\mathcal{D}$ that intersects $M_{i}$.
Suppose to the contrary that $|\mathcal{D}| > \mu$.
Then, there is $D \in \mathcal{D}$ such that no module $M_{i}$ is private for $D$.
If a module $M_{i}$ is adjacent to a module $M_{j}$ that intersects $D$,
then since $M_{j}$ is not private for $D$,
$\mathcal{D} \setminus \{D\}$ contains a unit intersecting $M_{j}$, which dominates $M_{i}$.
If a module $M_{i}$ intersects $D$, then 
since $D$ intersects at least two modules and $G[D]$ is connected,
there is a module $M_{j}$ adjacent to $M_{i}$ and intersecting $D$.
Hence, as the previous case, $\mathcal{D} \setminus \{D\}$ contains a unit dominating $M_{i}$.
Therefore, we can conclude that $\mathcal{D} \setminus \{D\}$ is an $r$-grouped dominating set.
This contradicts the minimality of $\mathcal{D}$.
\end{proof}

\subsection{Auxiliary subformulas}
\label{sec:subformulas}
Here we present \fo{} or \mso{2} expressions of some of the subformulas in Section~\ref{sec:logic}.
All of them are standard and presented only to show basic ideas.

The following formulas expressing dominating sets are almost direct translation 
of the definition and should be easy to read.
\begin{align*}
  \mathsf{dominating}(X)
  &= 
  \forall u \; \exists v \colon X(v) \land ((u = v) \lor E(u,v)).
  \\
  \mathsf{dominating}(v_{1}, \dots, v_{p})
  &= 
  \forall u \colon (u = v_{1}) \lor (u = v_{2}) \lor \dots \lor (u = v_{p}) \\
  & \qquad\quad \lor E(u, v_{1}) \lor E(u, v_{2}) \lor \dots \lor E(u, v_{p}).
\end{align*}

The connectivity of $G[X]$ is a little bit tricky to state.
We state that for each nonempty proper subset $Y$ of $X$, there is an edge between $Y$ and $X \setminus Y$.
See e.g., \cite{CyganFKLMPPS15} for the full expression of $\mathsf{connected}$.
The \fo{} version of $\mathsf{connected}$ can be expressed based on the same idea
but the length of the formula depends on the number $r$ of vertices it can take 
(which is fine for us as $r$ is part of the input when we use this formula).
In \cite{CyganFKLMPPS15}, an \mso{2} formula expressing the connectivity of the graph induced by an edge set is also presented.
We call it $\mathsf{connectedE}$ and use it below.

Recall that $\mathsf{span}(F,X)$ expresses that $X$ is the set of all endpoints of the edges in $F$
and that 
$\mathsf{cc}(F,C)$ expresses that $C$ is the vertex set of a connected component of the subgraph induced by $F$.
They can be stated as follows:
\begin{align*}
  \mathsf{span}(F,X) 
  &=
  \forall v \colon X(v) \Leftrightarrow (\exists e \colon F(e) \land I(e, v)),
  \\[.5ex]
  \mathsf{cc}(F,C) &=
  \exists F' \colon (F' \subseteq F) \land \mathsf{span}(F', C) \land \mathsf{connectedE}(F') \\
  & \qquad\quad \land (\forall F'' \colon (F' \subseteq F'' \subseteq F) \land \lnot \mathsf{connectedE}(F'')),
\end{align*}
where the inclusion relation $F \subseteq F'$ can be stated as $\forall e \colon F(e) \Rightarrow F'(e)$.

Finally, when $r$ is part of the parameter, 
$\mathsf{size}_{r}(C)$ meaning that $|C| = r$ can be stated as follows:
\begin{align*}
  \mathsf{size}_{r}(C) &= \dot{\exists} v_{1}, \dots, v_{r} \colon
  \bigwedge_{1 \le i \le r} C(v_{i}) \land \left(\lnot \exists v \colon C(v) \land \bigwedge_{1 \le i \le r} v \ne v_{i} \right).
\end{align*}